\documentclass[twocolumn,amsthm]{autart}
\usepackage{amsmath,amssymb,amsfonts,amscd}
\usepackage{graphicx}
\usepackage{tikz}
\usetikzlibrary{automata,positioning,arrows.meta,fit,backgrounds}
\usepackage{xspace}
\usepackage[all]{nowidow}
\usepackage[T1]{fontenc}
\usepackage[utf8]{inputenc}
\usepackage[sort,numbers]{natbib}
\usepackage{xcolor}
\usepackage{url}

\theoremstyle{plain}
\newtheorem{theorem}{Theorem}
\newtheorem{proposition}[theorem]{Proposition}
\newtheorem{lemma}[theorem]{Lemma}

\theoremstyle{definition}
\newtheorem{definition}[theorem]{Definition}

\theoremstyle{remark}

\newcommand{\So}{\Sigma_{o}}
\newcommand{\Suo}{\Sigma_{uo}}
\newcommand{\Shi}{\Sigma_{hi}}
\newcommand{\To}{T_{o}}
\newcommand{\Tuo}{T_{uo}}
\DeclareMathOperator{\Fill}{Fill}
\DeclareMathOperator{\Reach}{Reach}
\newcommand{\PTIME}{\textsc{P}\xspace}
\newcommand{\PSPACE}{\textsc{PSpace}\xspace}

\journal{arxiv.org}

\edef\endfrontmatter{%
        \unexpanded\expandafter{\endfrontmatter}
        \noexpand\endNoHyper 
}

\begin{document}

\begin{frontmatter}

\title{Complexity of Local Observation Consistency\\ in Discrete-Event Systems}

\author{Tom{\' a}{\v s} Masopust}\ead{tomas.masopust@upol.cz}

\address{Faculty of Science, Palack\'y University Olomouc, Czechia}

\begin{abstract}
	Hierarchical and multi-agent supervisory control under partial observation rely on \emph{local} consistency conditions between a plant and its abstraction: local observation consistency (LOC) for the projections of hierarchical control, and local relabeling observation consistency (LROC) for the relabelings of multi-agent systems. The complexity of verifying them has remained open. We formulate LOC for masks---abstractions that erase some events and rename the others and thus comprise both projections and relabelings---and show that its verification is \textsc{PSpace}-complete for nondeterministic plants, already for projections and for relabelings, and decidable in polynomial time for deterministic plants. The same holds for LROC, which is incomparable with LOC for relabelings.
\end{abstract}


\end{frontmatter}

\section{Introduction}\label{sec:intro}
	Hierarchical supervisory control of discrete-event systems (DES) controls a plant through its abstraction, and the central question is which conditions on the abstraction make synthesis on the abstract level correct and maximally permissive. Under complete observation, the answer is well known and easy to verify~\cite{ZhongWonham90,WongWonham96,SchmidtBreindl11}. Under partial observation, it is not. Boutin \emph{et al.}~\cite{BKMSS11} introduced two conditions for the projection abstraction: \emph{observation consistency} (OC)---a global condition relating the observations of the two levels---and \emph{local observation consistency} (LOC)---a local companion that plays for the abstraction the role the observer property plays under complete observation.

	The same pattern appears in multi-agent DES~\cite{LCL19,LKL22}, where isomorphic agents are collapsed onto a common template by a \emph{relabeling}. There the global condition is \emph{relabeling observation consistency} (ROC), and the local companion is \emph{local relabeling observation consistency} (LROC)~\cite{LKL22}.

	Projections erase events, and relabelings rename them. Both are instances of a \emph{mask}, which erases some events and renames the others; masks arise, e.g., when agents have internal events that the template does not need.

	The complexity of verifying the global conditions is known. For projections, OC and its asymmetric variant, \emph{modified observation consistency}, are undecidable already for deterministic plants, and become \PSPACE-complete under a structural restriction of the plant~\cite{MOC}; for relabelings, ROC is \PSPACE-complete~\cite{ROCPAPER}. The verification complexity of the local conditions, in contrast, has not been settled. For LOC of projections, we announced \PSPACE-completeness in~\cite{KM20} without proof; although a proof sketch appeared in the preprint of~\cite{KM20}, the proof given here is significantly simpler. For LROC, Liu \emph{et al.}~\cite{LKL22} observe that it is ``similar to a variant of observability and thus can be checked in a similar way'', which is correct for deterministic plants but leaves the case of nondeterministic plants open.

	We show that LOC verification for masks is in \PSPACE{} for nondeterministic plants and decidable in polynomial time for deterministic plants (Theorem~\ref{thm:loc-upper}), and that it is \PSPACE-hard already for projections and for relabelings (Propositions~\ref{prop:loc-hard} and~\ref{prop:locr-hard}). The polynomial bound for deterministic plants corrects an unproved claim in~\cite{KM20}. For relabelings, LOC reduces to an inclusion of languages, in which the observation plays no role, and it is incomparable with LROC (Example~\ref{ex:incomparable}). LROC verification is \PSPACE-complete for nondeterministic plants as well, and decidable in polynomial time for deterministic plants (Theorem~\ref{thm:lroc-complete}).

\section{Preliminaries}\label{sec:prelim}
	We assume that the reader is familiar with the basics of supervisory control~\cite{CassandrasLafortune08} and complexity theory~\cite{AB09}; \PTIME{} and \PSPACE{} denote the classes of problems decidable in polynomial time and in polynomial space.

	An \emph{alphabet} is a finite nonempty set of events. For an alphabet $\Sigma$, the set of finite strings over $\Sigma$ is denoted by $\Sigma^*$, and the empty string is denoted by $\varepsilon$. We write $A\mathbin{\dot\cup}B$ for the union of \emph{disjoint} sets $A$ and $B$, and we identify a singleton set $\{q\}$ with the element $q$. A language $L$ over $\Sigma$ is a subset of $\Sigma^*$. The \emph{prefix closure} of $L$ is the set $\overline{L}=\{u \in \Sigma^* : uv\in L \text{ for some } v \in \Sigma^*\}$, and $L$ is \emph{prefix-closed} if $L=\overline L$.

	A \emph{nondeterministic finite automaton} (NFA) is a tuple $G=(S,\Sigma,\delta,I,F)$, where $S$ is a finite set of states, $I \subseteq S$ is a set of initial states, $F\subseteq S$ is the set of final states, and $\delta\colon S\times\Sigma\to2^S$ is the transition function that can be extended to $2^S\times\Sigma^*$ in the usual way. Since every function is a relation, we also write $\delta\subseteq S\times\Sigma\times S$. The NFA is \emph{deterministic} (DFA) if $|\delta(q,a)|\le1$ for every $q\in S$ and $a\in \Sigma$, and $I=\{q_0\}$ for a single \emph{initial state} $q_0$. The language of $G$ is the set $L(G)=\{w \in \Sigma^* : \delta(I,w)\cap F\neq\emptyset\}$. Every plant here is prefix-closed, and therefore we take $F=S$, writing \emph{NFA with all states final}; in this case, we omit $F$ from the notation. The languages of such automata are the prefix-closed regular languages~\cite{KRS09}. Automata with $F\neq S$ arise only inside proofs.

	For a morphism $f$ and a language $K$, the \emph{inverse image of $K$ under $f$} is $f^{-1}(K)=\{s : f(s)\in K\}$. Under partial observation, the alphabet is partitioned as $\Sigma=\So\mathbin{\dot\cup}\Suo$ into observable and unobservable events, and the \emph{observation} $P\colon\Sigma^*\to\So^*$ is the morphism with $P(a)=a$ for $a\in\So$ and $P(a)=\varepsilon$ for $a\in\Suo$.

\begin{definition}\label{def:mask}
	Let $T=\To\mathbin{\dot\cup}\Tuo$ be an alphabet of observable and unobservable abstract events. A \emph{mask} is a morphism $M\colon\Sigma^*\to T^*$ given by a surjective map $M\colon\Sigma\to T\cup\{\varepsilon\}$ with $M(\So)\subseteq\To\cup\{\varepsilon\}$ and $M(\Suo)\subseteq\Tuo\cup\{\varepsilon\}$.
\end{definition}

	A mask \emph{erases} the events of $E=\{a\in\Sigma : M(a)=\varepsilon\}$ and renames the others, possibly merging several of them into one abstract event; for $\tau\in T$, we write $M^{-1}(\tau)=\{a\in\Sigma : M(a)=\tau\}$ for the set of events renamed to $\tau$. Two special cases are standard. In hierarchical control, a \emph{high-level alphabet} $\Shi\subseteq\Sigma$ is fixed, and the abstraction is the \emph{projection} $Q\colon\Sigma^*\to\Shi^*$, the mask with $T=\Shi$ and $E=\Sigma\setminus\Shi$, which merges no events. In multi-agent control, the abstraction is a \emph{relabeling} $R\colon\Sigma^*\to T^*$ onto a template alphabet $T$, the mask with $E=\emptyset$, which is letter-to-letter but in general not injective. The preservation of the observability status is part of both frameworks; the results on LOC below do not use it. The observation $P$ is itself a mask, onto $T=\To=\So$ with $E=\Suo$, and we use this in Section~\ref{sec:upper}.

\begin{definition}\label{def:loc}
	Let $M$ be a mask with erased events $E$, and let $\Sigma_c\subseteq\Sigma$ be the set of controllable events. A prefix-closed language $L\subseteq\Sigma^*$ is \emph{locally observation consistent} (LOC) with respect to $M$, $P$, and $\Sigma_c$ if for all $s,s'\in L$ with $P(s)=P(s')$ and every $e\in\Sigma_c\setminus E$ such that $M(s)M(e)\in M(L)$ and $M(s')M(e)\in M(L)$, there are \emph{fillers} $y,y'\in E^*$ with $P(y)=P(y')$, $sye\in L$, and $s'y'e\in L$.
\end{definition}

	Look-alike plant strings that can continue by a controllable event on the abstract level can do so on the plant after erased fillers that look the same. For the projection $Q$, this is the definition of Boutin \emph{et al.}~\cite[Def.~5]{BKMSS11}, with $e\in\Sigma_c\cap\Shi$ and fillers over $\Sigma\setminus\Shi$. For a relabeling $R$, the fillers are empty, and since $s'=s$ is allowed, the condition becomes a requirement on single strings: $L$ is LOC if and only if $L\Sigma_c\cap R^{-1}(R(L))\subseteq L$, that is, whatever controllable event the template allows next, the plant can execute; the observation plays no role. The language $R^{-1}(R(L))$ is prefix-closed and contains $L$, and the inclusion says that $L$ is controllable with respect to $R^{-1}(R(L))$ and $\Sigma_c$. In particular, LOC is unrelated to the controllability of a specification; it is a property of the plant and the relabeling. It is stronger than the observer property restricted to controllable events, which only asks that \emph{some} event $e''$ with $R(e'')=R(e)$ can follow $s$.

	The local condition of the multi-agent framework is defined only for relabelings~\cite{LKL22}.

\begin{definition}\label{def:lroc}
	A relabeling $R$ is \emph{locally relabeling observation consistent} (LROC) with respect to a prefix-closed language $L\subseteq\Sigma^*$ and the observation $P$ if for all $s,s'\in L$ with $P(s)=P(s')$ and all unobservable events $b,b'\in\Suo$ with $R(b)=R(b')$, if $sb\in L$ and $s'b'\in L$, then $s'b\in L$.
\end{definition}

	Both conditions are companions to a global condition that guarantees maximal permissiveness, but LOC transfers the enabling of a controllable event from the \emph{abstraction} to the \emph{plant}, whereas LROC transfers the enabling of merged unobservable events between look-alike \emph{plant} strings. For relabelings, they are incomparable.

\begin{exmp}\label{ex:incomparable}
	Let all events be unobservable, let $R(b)=R(c)$ with $b\neq c$, let $a$ be the only controllable event, with $R(a)$ distinct from $R(b)$, and let $L=\overline{\{b,ac\}}$. Then $L$ is LOC, since $\varepsilon$ is the only string $s$ with $R(s)R(a)\in R(L)$ and $a\in L$, but LROC fails for $s=\varepsilon$, $s'=a$, and the merged pair $(b,c)$, since $ab\notin L$. Conversely, let $R(e_1)=R(e_2)$ with $e_1\neq e_2$ observable and controllable, and let $L=\{\varepsilon,e_1\}$. Then LROC holds vacuously, as there are no unobservable events, but LOC fails for $s=s'=\varepsilon$ and the event $e_2$: the premise $R(\varepsilon)R(e_2)=R(e_1)\in R(L)$ holds, whereas $e_2\notin L$.
\end{exmp}

\section{Upper Bounds}\label{sec:upper}
	Throughout this section, $G=(S,\Sigma,\delta,I)$ is an NFA with all states final, $L=L(G)$, and $M$ is a mask with erased events $E$. For a set $X\subseteq S$ and $e\in\Sigma\setminus E$, we write
	\[
	  \Fill(X,e) = \bigl\{ P(y) : y\in E^*,\ \delta(X,ye)\neq\emptyset\bigr\}
	\]
	for the set of observations of the fillers that lead from $X$ to an occurrence of $e$, and $\Fill(s,e)=\Fill(\delta(I,s),e)$ for $s\in L$. It is the image under $P$ of the language of $G$ restricted to $E$-transitions, started in $X$, with the states carrying an outgoing $e$-transition made final; thus $\Fill(X,e)$ is regular, and an NFA for it with $|S|$ states is immediate from $G$. Since $sye\in L$ if and only if $\delta(\delta(I,s),ye)\neq\emptyset$, the fillers that LOC asks for exist if and only if two such sets meet:
	\begin{equation}\label{eq:loc-int}
		\begin{aligned}
			&\exists\,y,y'\in E^*\colon P(y)=P(y')\wedge sye\in L\wedge s'y'e\in L\\
			&\qquad\iff \Fill(s,e)\cap\Fill(s',e)\neq\emptyset \,.
		\end{aligned}
	\end{equation}
	The premise of LOC is a membership in the abstract language $M(L)$. The following lemma replaces it by a condition on pairs of plant strings with equal $M$-image, which can be tracked on $G$ (Lemma~\ref{lem:prod}). Since $M$ may merge events, the premise $M(s)M(e)\in M(L)$ may be witnessed by an event $e''\neq e$ with $M(e'')=M(e)$, and (C1) takes this into account.

\begin{lemma}\label{lem:loc-split}
	The language $L$ is LOC with respect to $M$, $P$, and $\Sigma_c$ if and only if both of the following conditions hold for every $e\in\Sigma_c\setminus E$:
	\begin{description}
		\item[(C1)] for all $s,z\in L$ with $M(s)=M(z)$ and every $e''\in\Sigma$ with $M(e'')=M(e)$, if $\Fill(z,e'')$ is nonempty, then so is $\Fill(s,e)$;
		\item[(C2)] for all $s,s'\in L$ with $P(s)=P(s')$, if $\Fill(s,e)$ and $\Fill(s',e)$ are both nonempty, then they intersect.
	\end{description}
\end{lemma}
\begin{proof}
	Assume that $L$ is LOC. To show (C1), let $s,z,e''$ be as there, and let $zye''\in L$ with $y\in E^*$. Then $M(s)M(e)=M(z)M(e'')=M(zye'')\in M(L)$, and LOC applied to $s'=s$ gives a filler $y'$ with $sy'e\in L$, so $\Fill(s,e)\neq\emptyset$. To show (C2), let $P(s)=P(s')$ with $\Fill(s,e)$ and $\Fill(s',e)$ nonempty. As above, with $z=s$ and $e''=e$, this gives $M(s)M(e)\in M(L)$, and likewise $M(s')M(e)\in M(L)$; LOC and~\eqref{eq:loc-int} yield the intersection.

	Conversely, assume (C1) and (C2), and let $s,s'\in L$ with $P(s)=P(s')$ and $M(s)M(e),M(s')M(e)\in M(L)$. Take $u\in L$ with $M(u)=M(s)M(e)$, and let $e''$ be the last letter of $u$ that is not erased. Then $u=ze''y$ with $y\in E^*$, $M(z)=M(s)$, and $M(e'')=M(e)$; since $L$ is prefix-closed, $ze''\in L$, and hence $\varepsilon\in\Fill(z,e'')$. Now (C1) gives $\Fill(s,e)\neq\emptyset$, and likewise $\Fill(s',e)\neq\emptyset$. Then (C2) and~\eqref{eq:loc-int} give the fillers required by LOC.
\end{proof}

	Both conditions quantify over pairs of strings with equal image under a mask, namely $M$ in (C1) and $P$ in (C2). Such pairs are tracked by the parallel composition~\cite{CassandrasLafortune08} of two relabeled copies of the subset automaton of $G$. The subset automaton $D$ of $G$ has the nonempty subsets of $S$ as states, the initial state $I$, and the transitions $X\xrightarrow{a}\delta(X,a)$ for $\delta(X,a)\neq\emptyset$. For a mask $N$ with erased events $E_N$, let $D_N$ be $D$ with every transition label $a$ replaced by $N(a)$, an automaton with $\varepsilon$-transitions, and let $G\otimes_N G=D_N\parallel D_N$. In the parallel composition, a transition labeled by an abstract event is taken by both components at once and an $\varepsilon$-transition by one of them, so that the transitions of $G\otimes_N G$ are, whenever the targets are nonempty,
	\begin{align*}
		(X,Y) &\to (\delta(X,a),\delta(Y,b)) && \text{if } N(a)=N(b)\neq\varepsilon,\\
		(X,Y) &\to (\delta(X,a),Y) && \text{if } a\in E_N,\\
		(X,Y) &\to (X,\delta(Y,a)) && \text{if } a\in E_N.
	\end{align*}
	For the observation, $P(a)=P(b)\neq\varepsilon$ means $a=b\in\So$, and for a projection $Q$, it means $a=b\in\Shi$. For an automaton $A$, let $\Reach_A$ denote the set of states of $A$ reachable from its initial states.

\begin{lemma}\label{lem:prod}
	For a mask $N$ and $X,Y\subseteq S$, $(X,Y)\in\Reach_{G\otimes_N G}$ if and only if there are $s,s'\in L$ with $N(s)=N(s')$, $X=\delta(I,s)$, and $Y=\delta(I,s')$.
\end{lemma}
\begin{proof}
	A state $(X,Y)$ of the parallel composition of two automata over the same alphabet is reachable by a string $t$ if and only if $X$ is reachable by $t$ in the first automaton and $Y$ in the second~\cite{CassandrasLafortune08}. In $D_N$, the state $X$ is reachable by $t$ if and only if $X=\delta(I,s)\neq\emptyset$ for some $s$ with $N(s)=t$, and since all states of $G$ are final, $\delta(I,s)\neq\emptyset$ if and only if $s\in L$.
\end{proof}

	By Lemma~\ref{lem:prod}, Lemma~\ref{lem:loc-split} reads as follows: for every $e\in\Sigma_c\setminus E$,
	\begin{align}
		\forall (X,Y)\in\Reach_{G\otimes_M G}\ \forall e''\in M^{-1}(M(e)) \colon & \notag\\
		\Fill(Y,e'') \neq \emptyset \Rightarrow \Fill(X,e) &\neq \emptyset\,, \label{eq:C1}\\
		\forall (X,Y)\in\Reach_{G\otimes_P G} \colon \Fill(X,e) \neq \emptyset, \Fill&(Y,e) \neq \emptyset \notag\\
		\Rightarrow \Fill(X,e) \cap \Fill(Y,e) &\neq \emptyset \,. \label{eq:C2}
	\end{align}

	The LOC verification problem asks, given an NFA $G$ over $\Sigma$, a partition $\Sigma=\So\mathbin{\dot\cup}\Suo$, a mask $M$, and $\Sigma_c\subseteq\Sigma$, whether $L(G)$ is LOC with respect to $M$, $P$, and $\Sigma_c$.

\begin{theorem}\label{thm:loc-upper}
	LOC verification is in \PSPACE, and it is decidable in polynomial time for DFAs.
\end{theorem}
\begin{proof}
	For $X\subseteq S$ and $e\in\Sigma\setminus E$, an NFA for $\Fill(X,e)$ with $|S|$ states is obtained from $G$ by keeping only the transitions under $E$, taking $X$ as the set of initial states, declaring final the states with an outgoing $e$-transition, and relabeling each transition by its observation. Hence, for given $X,Y\subseteq S$ and events $e,e''$, emptiness of $\Fill(X,e)$ and of $\Fill(X,e)\cap\Fill(Y,e)$ is decidable in polynomial time, by reachability in the NFA and in the product of two such NFAs, and so is whether $(X,Y)$ violates~\eqref{eq:C1}, that is, $\Fill(Y,e'')\neq\emptyset$ and $\Fill(X,e)=\emptyset$, or~\eqref{eq:C2}, that is, $\Fill(X,e)$ and $\Fill(Y,e)$ are nonempty and disjoint.

	By Lemma~\ref{lem:loc-split} and Lemma~\ref{lem:prod}, $L$ is not LOC if and only if, for some $e\in\Sigma_c\setminus E$, a state of $\Reach_{G\otimes_M G}$ violates~\eqref{eq:C1} for some $e''\in M^{-1}(M(e))$, or a state of $\Reach_{G\otimes_P G}$ violates~\eqref{eq:C2}. A nondeterministic procedure guesses $e$, $e''$, and the violated condition, then guesses a path in $G\otimes_M G$, respectively $G\otimes_P G$, from $(I,I)$, storing only the current state, which takes $2|S|$ bits, and finally tests the reached state for the violation. This runs in nondeterministic polynomial space, and $\textsc{NPSpace}=\PSPACE$ by Savitch's theorem.

	If $G$ is a DFA with $n$ states, the states of $G\otimes_M G$ and $G\otimes_P G$ reachable from $(q_0,q_0)$ are pairs of singletons, so that $\Reach_{G\otimes_M G}$ and $\Reach_{G\otimes_P G}$ have at most $n^2$ elements and are computed by reachability; testing each of them for a violation is polynomial as well.
\end{proof}

	For the two special cases, Lemma~\ref{lem:loc-split} simplifies in opposite ways. A projection merges no events, so $e''=e$ in~(C1), and both conditions remain. For a relabeling, $E=\emptyset$, so $\Fill(X,e)$ is $\{\varepsilon\}$ if some state of $X$ has an outgoing $e$-transition and $\emptyset$ otherwise; then~(C2) always holds, and~(C1) is the inclusion $L\Sigma_c\cap R^{-1}(R(L))\subseteq L$.

	Theorem~\ref{thm:loc-upper} corrects a claim made in~\cite{KM20}, where it is asserted without proof that LOC verification for projections is not easier for deterministic than for nondeterministic plants.

\section{Lower Bounds}\label{sec:lower}
	The following lower bound for projections is stated in~\cite[Thm.~6]{KM20} without proof. We provide a proof, significantly simplified compared with the proof sketch in the preprint of~\cite{KM20}, and improve the result.

\begin{proposition}\label{prop:loc-hard}
	LOC verification is \PSPACE-hard for NFAs already for projections, even if $\Sigma_c=\So=\Shi$.
\end{proposition}
\begin{proof}
	We reduce from the universality of NFAs with all states final, which is \PSPACE-complete~\cite{KRS09}. Let $A$ be an NFA over $\Sigma_A$ with all states final. Then $L(A)$ is prefix-closed, and we may assume that $\varepsilon\in L(A)$; otherwise, $L(A)=\emptyset$ and the instance is trivial. Let $\Sigma_A'=\{a' : a\in\Sigma_A\}$ be a disjoint copy of $\Sigma_A$, and set
	\[
	  \Sigma = \Sigma_A\cup\Sigma_A' \quad \text{ and } \quad \Sigma_c = \So = \Shi = \Sigma_A,
	\]
	that is, $Q=P$ erases the primed letters and the alphabet $\Sigma\setminus\Shi=\Sigma_A'$ is entirely unobservable. From $A$, we construct an NFA $B$ over $\Sigma$, with all states final, as follows: $B$ has a fresh initial state $p_0$, a fresh state $p_1$, a copy of $A$, and two fresh states $r_0,r_1$; it has the transitions $(p_0,a,p_1)$ and $(p_1,a,i)$ for every $a\in\Sigma_A$ and every initial state $i$ of $A$, together with all transitions of the copy of $A$, and the transitions $(p_0,a,r_1)$, $(r_1,a',r_0)$, and $(r_0,a,r_1)$ for every $a\in\Sigma_A$. Thus
	\[
	  L = L(B) = \overline{\Sigma_A\Sigma_A L(A)} \;\cup\; \overline{(\Sigma_A\Sigma_A')^*}\,,
	\]
	the first branch goes through states $p_0,p_1$ and the copy of $A$, called the \emph{$A$-branch}, and the second goes through states $r_0,r_1$, called the \emph{alternating branch}; see Fig.~\ref{fig:B}. The automaton $B$ has $|A|+4$ states and is constructed from $A$ in time polynomial in the size of $A$.

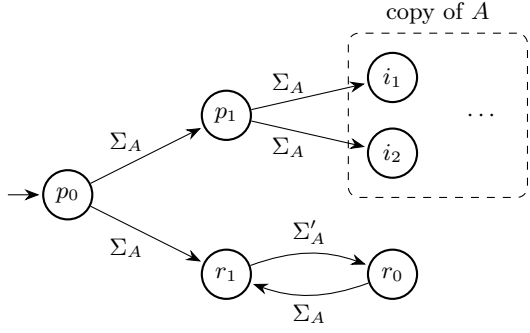
\begin{figure}[t]
\centering
\begin{tikzpicture}[shorten >=1pt,
   every state/.style={minimum size=6.5mm, inner sep=1pt, thick},
   >={Stealth[length=2mm]}, initial text={}, font=\small]
 
  \node[state, initial left] (p0) at (0,0)      {$p_0$};
  \node[state]               (p1) at (2.1,1.05) {$p_1$};
  \node[state]               (i1) at (4.3,1.55) {$i_1$};
  \node[state]               (i2) at (4.3,0.55) {$i_2$};
  \node[font=\footnotesize] (dots) at (5.5,1.05) {$\cdots$};
  \node[state]               (r1) at (2.1,-1.05) {$r_1$};
  \node[state]               (r0) at (4.3,-1.05) {$r_0$};
 
  \path[->]
    (p0) edge node[above left=-3pt] {$\Sigma_A$}  (p1)
    (p0) edge node[below left=-3pt] {$\Sigma_A$}  (r1)
    (p1) edge node[above left=-3pt] {$\Sigma_A$}  (i1)
    (p1) edge node[below left=-3pt] {$\Sigma_A$}  (i2)
    (r1) edge[bend left=20] node[above]            {$\Sigma_A'$} (r0)
    (r0) edge[bend left=20] node[below]            {$\Sigma_A$}  (r1);
 
  \begin{scope}[on background layer]
    \node[draw, dashed, rounded corners, fit=(i1)(i2)(dots), inner sep=7pt] (Abox) {};
  \end{scope}
  \node[font=\small, anchor=south] at (Abox.north) {copy of $A$};
\end{tikzpicture}
\caption{The automaton $B$ of Proposition~\ref{prop:loc-hard}. The upper
part is the $A$-branch, the lower part is the alternating branch.}
\label{fig:B}
\end{figure}

	Since the alphabet $\Sigma_A'$ is unobservable, $P(u)=\varepsilon$ for every $u \in \Sigma_A'^*$, and LOC is equivalent to the condition
	\begin{multline}\label{eq:loc-single}
		  \forall s\in L\ \forall e\in\Sigma_c\cap\Shi\colon\\
		  Q(s)e\in Q(L) \implies \exists u\in(\Sigma\setminus\Shi)^*\colon sue\in L\,.
	\end{multline}
	Indeed, LOC with $s'=s$ gives~\eqref{eq:loc-single}.
	Conversely, suppose~\eqref{eq:loc-single} and let $s,s'\in L$ and $e\in \Sigma_c \cap \Shi$ be such that $P(s)=P(s')$ and $Q(s)e,Q(s')e\in Q(L)$. Applying~\eqref{eq:loc-single} to $s$ and $s'$ separately yields fillers $u$ and $u'$ with $sue\in L$ and $s'u'e\in L$. Since the filler alphabet $\Sigma \setminus \Shi$ is unobservable, the two fillers meet the requirement $P(u)=\varepsilon=P(u')$ of LOC.

	It remains to show that~\eqref{eq:loc-single} holds if and only if $L(A)=\Sigma_A^*$.

	Assume that $L(A)=\Sigma_A^*$, and let $s\in L$ and $e\in\Sigma_A$. If $s$ is generated along the alternating branch, then $B$ is in state $p_0$, $r_0$, or $r_1$ after $s$. In the first two cases, $se\in L$, and $u=\varepsilon$ is the required filler; in the third case, $sa'e\in L$ for every $a'\in\Sigma_A'$, and hence $u=a'$ is the required filler. Otherwise, $s$ is generated only along the $A$-branch, which is the case only if $|s|\ge2$, and then $s=abt$ with $a,b\in\Sigma_A$ and $t\in L(A)$. Since  $L(A)=\Sigma_A^*$, we have $te\in L(A)$, and hence $se=abte\in L$ for $u=\varepsilon$. In all cases~\eqref{eq:loc-single} holds.

	Conversely, assume that $L(A)\neq\Sigma_A^*$. Since $\varepsilon \in L(A)$, there are $t\in L(A)$ and $e\in\Sigma_A$ with $te\notin L(A)$. For $a,b\in\Sigma_A$, the string $s=abt \in L$ is generated along the $A$-branch. However, no filler $u$ satisfies $sue\in L$. Indeed, since $s$ contains two consecutive unprimed letters, every string of $L$ extending $s$ is generated along the $A$-branch. After $s$, the automaton $B$ is in a state of the copy of $A$, from which no primed letter is available, and therefore $u \in \{\varepsilon\}$. However, $se=abte\in L$ would imply $te\in L(A)$, which is a contradiction. Since $Q(L)=\Sigma_A^*$, we have $Q(s)e\in Q(L)$ and~\eqref{eq:loc-single} does not hold.
\end{proof}

	The reduction of Proposition~\ref{prop:loc-hard} relies on fillers, which relabelings do not have. For relabelings, a merged copy of the alphabet takes their place.

\begin{proposition}\label{prop:locr-hard}
	LOC verification is \PSPACE-hard for NFAs already for relabelings, even if all events are unobservable and only one event is controllable.
\end{proposition}
\begin{proof}
	We reduce the universality of NFAs with all states final over $\Delta=\{a,b\}$, which is \PSPACE-complete~\cite{KRS09}. Let $A$ be such an NFA, let $\Delta'=\{a',b'\}$ be a disjoint copy of $\Delta$, and let $\Sigma=\Delta\mathbin{\dot\cup}\Delta'\mathbin{\dot\cup}\{e\}$ with $\Sigma_c=\{e\}$ and $\So=\emptyset$. Let $R$ map $a$ and $a'$ to $\hat a$, $b$ and $b'$ to $\hat b$, and $e$ to $\hat e$, and consider
	\[
		L \;=\; \Delta^{*}\ \cup\ \overline{L(A)\,e}\ \cup\ \overline{\Delta'^{*}e} ,
	\]
	which is prefix-closed and generated by an NFA with all states final and $|A|+3$ states. For a relabeling, $L$ is LOC if and only if $se\in L$ for every $s\in L$ with $R(s)\hat e\in R(L)$. For every $s\in\Delta^*$, the string $R(s)\hat e$ lies in $R(L)$, being the image of the copy of $s$ over $\Delta'$ followed by $e$, and $se\in L$ if and only if $s\in L(A)$. For $s\in\Delta'^*$, $se\in L$ holds, and for a string $s$ that contains $e$, $R(s)\hat e\notin R(L)$. Hence $L$ is LOC if and only if $L(A)=\Delta^*$.
\end{proof}

	Since universality of NFAs with all states final is \PSPACE-complete already over a binary alphabet~\cite[Lem.~2]{KRS09}, the reduction of Proposition~\ref{prop:loc-hard} uses four events, two observable and two unobservable, and that of Proposition~\ref{prop:locr-hard} five unobservable events. Together with Theorem~\ref{thm:loc-upper}, we obtain the following.

\begin{theorem}\label{thm:loc-complete}
	LOC verification is \PSPACE-complete for NFAs, already for projections and for relabelings, and decidable in polynomial time for DFAs. \qed
\end{theorem}

	Neither lower bound implies the other: a relabeling cannot erase the fillers of Proposition~\ref{prop:loc-hard}, and a projection cannot merge the events of Proposition~\ref{prop:locr-hard}.

	Theorem~\ref{thm:loc-complete} rules out a polynomial-time algorithm for nondeterministic plants unless $\PTIME=\PSPACE$, but a polynomial-time sufficient condition is available. Since $\Fill(X,e)=\bigcup_{q\in X}\Fill(q,e)$, conditions~\eqref{eq:C1} and~\eqref{eq:C2} hold whenever they are satisfied for single states, that is, with $\Reach_{G\otimes_M G}$ and $\Reach_{G\otimes_P G}$ replaced by $\Reach_{G_M\parallel G_M}$ and $\Reach_{G_P\parallel G_P}$, where $G_N$ is $G$ with every transition label $a$ replaced by $N(a)$; indeed, by the argument of Lemma~\ref{lem:prod} applied to $G$ instead of $D$, $\Reach_{G_N\parallel G_N}=\bigcup_{(X,Y)\in\Reach_{G\otimes_N G}}X\times Y$. These sets have at most $|S|^2$ elements, so that the condition is decidable in polynomial time, and for DFAs it coincides with~\eqref{eq:C1} and~\eqref{eq:C2}. For $N=P$, $G_P\parallel G_P$ is the product used for the polynomial verification of observability and diagnosability~\cite{Tsitsiklis89,YooLafortune02}.

\section{The Local Condition of Multi-Agent Control}\label{sec:lroc}
	LROC is a condition on unobservable events, and its verification is similar.

\begin{proposition}\label{prop:lroc-ptime}
	LROC verification is in \PSPACE, and decidable in polynomial time for DFAs.
\end{proposition}
\begin{proof}
	Let $G=(S,\Sigma,\delta,I)$ be an NFA with all states final. By Definition~\ref{def:lroc}, LROC fails if and only if there are $s,s'\in L=L(G)$ with $P(s)=P(s')$ and $b,b'\in\Suo$ with $R(b)=R(b')$ such that $sb\in L$, $s'b'\in L$, and $s'b\notin L$. Since all states of $G$ are final, these three memberships are determined by the state sets $X=\delta(I,s)$ and $X'=\delta(I,s')$: they say $\delta(X,b)\neq\emptyset$, $\delta(X',b')\neq\emptyset$, and $\delta(X',b)=\emptyset$. We call a pair $(X,X')$ of subsets of $S$ \emph{bad} if some $b,b'\in\Suo$ with $R(b)=R(b')$ satisfy these three conditions. Whether a given pair is bad is decided by inspecting at most $|\Suo|^2$ choices of $(b,b')$ and applying $\delta$, in time polynomial in the size of $G$.

	By Lemma~\ref{lem:prod} applied with $N=P$, the pairs $(X,X')$ that arise from strings $s,s'\in L$ with $P(s)=P(s')$ are exactly the elements of $\Reach_{G\otimes_P G}$. Hence LROC fails if and only if $\Reach_{G\otimes_P G}$ contains a bad pair. A nondeterministic procedure guesses a path in $G\otimes_P G$ from $(I,I)$, storing only the current state, which takes $2|S|$ bits, and tests the reached state for badness. This runs in nondeterministic polynomial space, and $\textsc{NPSpace}=\PSPACE$ by Savitch's theorem.

	If $G$ is a DFA, the states of $G\otimes_P G$ reachable from $(q_0,q_0)$ are pairs of singletons, so that $\Reach_{G\otimes_P G}$ has at most $|S|^2$ elements and is computed by reachability. Testing each of its elements for badness is polynomial, and hence the whole procedure is polynomial.
\end{proof}

	The polynomial bound makes the observation of~\cite{LKL22} that LROC ``can be checked in a similar way as observability'' precise. It does not survive nondeterminism.

\begin{proposition}\label{prop:lroc-hard}
	LROC verification is \PSPACE-hard for NFAs, even for a relabeling injective on observable events that merges one pair of unobservable events.
\end{proposition}
\begin{proof}
	We reduce universality of NFAs with all states final over $\Delta=\{a,b\}$, which is \PSPACE-complete~\cite{KRS09}. Let $A$ be such an NFA; then $L(A)$ is prefix-closed. Let $\Sigma=\Delta\mathbin{\dot\cup}\{c,d_1,d_2\}$ 
	with $\So=\Delta$ and $\Suo=\{c,d_1,d_2\}$, and let $R$ be injective on $\Delta$, with $R(c)=\hat c$ and $R(d_1)=R(d_2)=\hat d$. Consider the language
	\[
		L \;=\; \overline{c\,\Delta^{*}d_1}\ \cup\ \overline{c\,\Delta^{*}d_2}\ \cup\ \overline{\Delta^{*}d_2}\ \cup\ \overline{L(A)\,d_1} ,
	\]
	which is prefix-closed and can obviously be generated by an NFA with all states final of size polynomial in the size of $A$.

	Since $R$ is injective on $\Sigma$ except for the pair $d_1,d_2$, the only instances of Definition~\ref{def:lroc} with $b\neq b'$ are $(b,b')=(d_1,d_2)$ and $(b,b')=(d_2,d_1)$. Observations erase $c,d_1,d_2$, and $P(w)=P(cw)=w$ for $w\in\Delta^{*}$.

	For $(b,b')=(d_2,d_1)$, $s' d_1 \in L$ implies $s'\in c \Delta^* \cup L(A)$, that is, $s'=cw$ or $s'=w$ for some $w\in\Delta^*$; since both $cwd_2$ and $wd_2$ belong to $L$, we have $s'd_2\in L$, and this pair never violates LROC. For $(b,b')=(d_1,d_2)$, $s' d_2 \in L$ restricts $s'$ to be of the form $w$ or $cw$ for some $w\in \Delta^*$. If $s'=cw$, then $s'd_1=cwd_1\in L$ and there is no violation. If $s'=w$, then $s'd_2=wd_2\in L$, the string $s=cw$ satisfies $P(s)=P(s')$ and $sd_1=cwd_1\in L$, and $s'd_1=wd_1 \in L$ if and only if $w\in L(A)$. A violation therefore exists if and only if there is $w\in\Delta^{*} \setminus L(A)$, that is, if and only if $L(A)\neq\Delta^{*}$.
\end{proof}

	We immediately have the following.
\begin{theorem}\label{thm:lroc-complete}
	LROC verification is \PSPACE-complete for NFAs and decidable in polynomial time for DFAs.
\end{theorem}

\section{Conclusion}\label{sec:concl}
	The local conditions of hierarchical and multi-agent supervisory control are \PSPACE-complete to verify for nondeterministic plants, and they become polynomial in time as soon as the plant is deterministic. For LOC, this holds for every mask, and the lower bound holds already for projections and for relabelings; for relabelings, LOC is an inclusion of languages and is incomparable with LROC. For nondeterministic plants, a polynomial-time sufficient condition may precede the exponential search.

\section*{Declaration of generative AI and AI-assisted technologies in the writing process}
	During the preparation of this work, the author used Claude (Anthropic) to edit the language and presentation, to restructure the manuscript around masks, to draft and simplify parts of the proofs, including the extension of the upper bounds from projections to masks and the lower bound for relabelings, and to check the proofs and examples computationally on randomly generated data. The author reviewed and edited the content as needed and takes full responsibility for the content of the publication.

\bibliographystyle{elsarticle-num}
\bibliography{loc}

\end{document}